\documentclass{IEEEtran}

\usepackage{balance}
\usepackage{amsmath,graphicx}
\usepackage{amsfonts}
\usepackage{color}
\usepackage{amssymb}
\usepackage{amsthm}
\usepackage{bm}
\usepackage{cite}

\newtheorem{lemma}{Lemma}

\newtheorem{remark}{Remark}
\usepackage{cuted}
\usepackage{array,multirow}
\usepackage{stfloats}
\usepackage{bbm}
\usepackage{mathtools}
\usepackage{algorithm,algorithmic} %
\usepackage{multicol}

\newcommand{\revisNew}{\textcolor{black}}

\newcommand{\be}{\begin{equation}}
\newcommand{\ee}{\end{equation}}
\newcommand{\ist}{\hspace*{.4mm}}
\newcommand{\rmv}{\hspace*{-.4mm}}

\begin{document}

\title{Arithmetic Average Density Fusion - Part III: Heterogeneous Unlabeled and Labeled RFS Filter Fusion 
}

\author{Tiancheng~Li,~\IEEEmembership{Senior Member,~IEEE}, Ruibo Yan, Kai Da and Hongqi Fan
\thanks{
This work was partially supported 
 by National Natural Science Foundation of China (Grant No. 62071389), Natural Science Basic Research Program of Shaanxi Province (Program No. 2023JC-XJ-22), 
 and Key Laboratory Foundation of National Defence Technology (No. JKWATR-210504). 
}
\thanks{T.\ Li and R. Yan are with the Key Laboratory of Information Fusion Technology (Ministry of Education), School of Automation, Northwestern Polytechnical University, Xi'an 710129, China, e-mail: t.c.li@nwpu.edu.cn, 2018300490@mail.nwpu.edu.cn%
}
\thanks{K. Da and H. Fan are with the Key Laboratory of Science and Technology on ATR, National University of Defense Technology, Changsha 410073, China.
e-mail: dktm131@163.com, fanhongqi@nudt.edu.cn}
}

\maketitle

\begin{abstract}
This paper proposes a heterogenous density fusion approach to scalable multisensor multitarget tracking 
where the inter-connected sensors run different types of random finite set (RFS) filters according to their respective capacity and need. {These diverse RFS filters} result in heterogenous multitarget densities that are to be fused with each other in a proper means for more robust and accurate detection and localization of the targets.
Our approach is based on Gaussian mixture implementations where the local Gaussian components (L-GCs) are {revised for PHD consensus, i.e., the corresponding unlabeled probability hypothesis densities (PHDs) of each filter best fit their average} regardless of the specific type of the local densities. To this end, a computationally efficient, \revisNew{coordinate descent} approach is proposed which only revises the weights of the L-GCs, keeping the other parameters 
unchanged. 
In particular, the PHD filter, the unlabeled and labeled multi-Bernoulli (MB/LMB) filters are considered. 
Simulations have demonstrated the effectiveness of the proposed approach for both homogeneous and heterogenous fusion of the PHD-MB-LMB filters in different configurations. 
\end{abstract}




\section{Introduction}

With the proliferation of internet of things \cite{Qiu18HeteroIoT}, heterogenous sensor networks 
become {promising for multisensor multitarget (MSMT) tracking} where the sensors have unequal computing and memory capacities {and/or different measurement models} and correspondingly run their own suitable algorithms.
In fact, it is {often required to operate} different algorithms even in a homogeneous sensor network, to the prejudice of the network reliability and lifetime \cite{Yarvis05hetero, Liggins17}. In either case, the local filters are allowed to apply different statistical models (regarding the target birth, detection, survival, and death and the clutter) and {approximation} algorithms. Thanks to the diversity, the combination {of heterogenous filters} is naturally more robust and reliable as compared with the unitary one.
In this paper, we consider such a {heterogenous MSMT tracking problem} where (whether homogenous or heterogenous) netted sensors run different types of multitarget filters/trackers {derived by using different families of random finite sets (RFSs)} such as unlabeled and labeled RFS filters \cite{Mahler14book,Vo15mtt}. {Our goal is to find a technically solid means to cooperate these cooperative filters via inter-sensor information fusion for improving their respective, heterogenous estimation.} 

One solution to the above heterogenous filter cooperation could be given by inter-sensor sharing the raw measurements so that the heterogenous filters just use the aggregated measurements of different sensors. This {is particularly useful for dealing with the lack of observability \cite{Dimitri21,Xie18} which, however, }
does not suit the large-scale peer-to-peer networks. 
More importantly, it is computationally intractable 
to make the optimal use of the measurements of all sensors due to the explosive possibility for track-to-measurement association \cite{Mahler14book,Delande11MsPHD,Nannuru16MsCPHD,Wei16centLMB,Saucan17MsMB,Vo19msGLMB,Si20msPMBM,Robertson22MsLMB,Trezza22Msbirth,Moratuwage22MSMoT}.
Differently, we resort to the computationally efficient and scalable density average consensus approach \cite{Li22chapter} to fuse the probability hypothesis density (PHD) filter \cite{Vo05,Vo06}, the unlabeled multiple Bernoulli (MB) filter \cite{Vo09CBmember} and the labeled MB (LMB) filter \cite{Vo13Label,Reuter14LMB}. 
{It requires only local data communications and simple fusion computation and is proven resilient to complicated/unknown inter-sensor correlation \cite{Bailey12,Li20AAmb}.}
According to the best of our knowledge, all existing density fusion approaches \cite{Li22chapter,Li22RFS-AA-Derivation} are homogeneous {in the parameter family of the fusing densities}, i.e., fusion is carried out among the same type of filters. Note that the defined heterogeneous fusion is different from 
what given in
\cite{Petitti11HeterEstimates,Dagan20HeterogeneousChannel,Yi21Heterogeneous,Arulampalam21HeteBearing} where the filters to be fused are still the same type. 


Our proposal in this paper is based on
the arithmetic-average (AA) density fusion \cite{Li2021SomeResults,Li22RFS-AA-Derivation}, {which has recently been tailored to accommodate different RFS densities all demonstrating significant advantages in dealing with misdetection and in high computing efficiency}. 
Effective AA fusion implementations of either unlabeled or labeled RFS filters are all derived from the same (labeled) PHD-AA formulation which ensures consistency in the (labeled) PHD estimation \cite{Li22RFS-AA-Derivation} {and ultimately} more accurate and robust detection and localization of the present targets.
This does not only expose the fact that existing linear-fusion-based multi-sensor RFS filters merely seek consensus over the (labeled) PHDs rather than the multi-object probability densities (MPDs) but also paves a way for heterogenous RFS filter cooperation via averaging their respective unlabled/labeled PHDs as shown in Fig. \ref{fig:HFframework}. In view of this, this paper presents the first ever heterogenous unlabeled and labeled RFS filter cooperation approach based on 
the Gaussian mixture (GM) implementation. 


%

In this paper, the sensors are assumed connected with each other either via a centralized fusion center or via a peer-to-peer distributed network. In the latter, the popular average consensus approach \cite{Olfati-Saber07,Sayed14book} or the distributed flooding approach \cite{Li17flooding} can be used for inter-node communication. The AA density fusion can be easily applied in both cases. {Nevertheless, this paper focuses on the centralized fusion for brevity}, omitting inter-node communication issues.
We further assume that the 
sensors are synchronous and accurately coordinated, whose fields of view cover the region of interest (ROI). 
{The GM implementations of the local PHD \cite{Vo06}, MB \cite{Vo09CBmember} and LMB \cite{Reuter14LMB} filters are the standard as given in the references in which details are available}. 
These limitations can be relaxed to extend our proposed approach. 

\begin{figure}
  \centering
  \includegraphics[width=4.5cm]{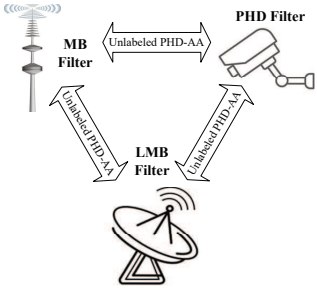}\\  
  \caption{A heterogenous RFS filter cooperation scenario based on unlabeled PHD-AA fusion} \label{fig:HFframework}
  \vspace{-2mm}
\end{figure}

The paper is organized as follows. A brief introduction to the Poisson, unlabeled and labeled MBs and their GM implementations are given in sections \ref{sec:background} and \ref{GM-implementation}, respectively. Their heterogenous fusion performed by merely revising the weights of the local Gaussian components (L-GCs) is given in section \ref{sec:GMfit}. 
Simulation is given in section \ref{sec:simulation} before the paper is concluded in section \ref{sec:conclusion}.

%

\section{Preliminaries: PHD, MB and LMB} \label{sec:background}
\subsection{Basic MSMT Scenario}
The following scenario assumptions are made in this paper. Target births are independent of target survivals, which might be modelled differently at local sensors based on either Poisson or MB RFSs. Each target evolves and generates measurements independently. 
At time $k-1$, the target with state $\mathbf{x}_{k-1} \in \mathcal{X}$ will either die with probability $1-p^{\text{s}}_k(\mathbf{x}_k)$ or persists at time $k$ with survival probability $p^{\text{s}}_k(\mathbf{x}_k)$ and attains a new state $\mathbf{x}_k$ according to a Markov transition probability density function (PDF) $f_{k|k-1} (\mathbf{x}_k|\mathbf{x}_{k-1})$. Hereafter, $\mathcal{X} \subseteq \mathbb{R}^d$ denotes the $d$-dimensional state space.

Given a target with state $\mathbf{x}_k \in \mathcal{X}$
, sensor $i \in \mathcal{I}$ either detects it with probability $p^{\text{d}}_{i,k}(\mathbf{x}_k)$ 
and generates a measurement $\mathbf{z}_{i,k}\in Z_{i,k}$  with likelihood $g_{i,k}(\mathbf{z}_{i,k}|\mathbf{x}_k)$ or fails to detect it with probability $1-p^{\text{d}}_{i,k}(\mathbf{x}_k)$, where $Z_{i,k}$ denotes the {set} of the measurements received at time $k$ by sensor $i$.
The clutter (namely the measurement of no target) follows a Poisson RFS, independent of target measurements. 

\subsection{RFS Modeling}
The states of a random number of targets are described by a RFS ${X} = \big\lbrace \mathbf{x}_{1}, \dots, \mathbf{x}_{n} \big\rbrace \in \mathbb{X}$, where $n =|{X}|$ denotes the random number of targets and $\mathbb{X}$ denotes all of the finite subsets of $\mathcal{X}$.
For any realization of $X$ with a given cardinality $|{{X}}| =n$, namely ${X}_n = \big\lbrace \mathbf{x}_{1}, \dots, \mathbf{x}_{n} \big\rbrace$, \cite[Eq.2.36]{Mahler14book},
the MPD is defined as
$f({X}_n) = n! \rho(n) f(\mathbf{x}_{1}, \dots, \mathbf{x}_{n})$
where the localization densities $f(\mathbf{x}_{1}, \dots, \mathbf{x}_{n} )$ are symmetric in their arguments and the cardinality distribution is given by $\rho(n)\triangleq \mathrm{Pr}\{|{{X}}|=n\} = \int_{|{X}| = n} {f({X})\delta {X}}$.
The set integral in $\mathbb{X}$ is defined as \cite[Ch. 3.3]{Mahler14book}
$\int_{\mathbb{X}} {f({X})\delta {X}}
   = 
\sum_{n = 0}^\infty  {\frac{1}{{n!}}\int_{\mathcal{X}^n} {f\big(\{ {\mathbf{x}_1},\dots,{\mathbf{x}_n}\} \big)d{\mathbf{x}_1}\dots d{\mathbf{x}_n}}}
$ 
where $f(\emptyset )= \rho(0)$.

The PHD $D(\mathbf{x})$, also known as the first moment density \cite[pp. 168-169]{Goodman97},
of the multitarget density $f({X})$ is a density function on single target $\mathbf{x}\in X$, defined as \cite[Ch.4.2.8]{Mahler14book} 
\begin{align}
D(\mathbf{x})   
& \triangleq \int_\mathbb{X}   {\bigg(\sum_{\mathbf{y}\in X}{\delta_\mathbf{y}}(\mathbf{x}) \bigg)f(X)\delta X} \label{def-PHD}
\end{align}
where \revisNew{${\delta_\mathbf{y}}(\mathbf{x})$ denotes the Dirac delta function concentrated at $\mathbf{y}$.} 

The PHD is the density of the expected cardinality w.r.t. hyper-volume, which
has clear physical significance as its integral in any region $\mathcal{S} \subseteq \mathcal{X} $ gives the expected number $\hat{N}^{\mathcal{S}}$ of targets in that region, i.e.,
$  \hat{N}^{\mathcal{S}} = \int_{\mathcal{S}} D(\mathbf{x}) d \mathbf{x}$.
Arguably, it tells how well the present targets are detected in the local region. The PHD plays a key role in all RFS filters and undoubtedly, a good PHD estimate implies a good filter estimate.
To distinguish the targets, one needs to use the labeled RFS (LRFS) which is a RFS whose elements are assigned with distinct labels \cite{Vo13Label,Vo14GLMB}. 
Denote by {$\mathbb{L}$ all of the finite subsets of the label space.} 
A realization of a LRFS with cardinality $n$, multitarget state $X_n$ and label set $L_n = \big\lbrace l_{1}, l_{2}, \dots, l_{n} \big\rbrace \in \mathbb{L}$ is denoted by $\widetilde{X}_n = \{ (\mathbf{x}_1, l_1),(\mathbf{x}_2, l_2),...,(\mathbf{x}_n, l_n) \} \in \mathbb{X} \times \mathbb{L}$. The LRFS is completely characterized by its multitarget density $\pi\big(\widetilde{X}\big)$. 
Consequently, 
the labeled and unlabeled PHD for a labeled RFS are respectively given as follows \cite{Vo13Label}
\begin{align}
  {\widetilde D} (\mathbf{x},l) & \triangleq \int_{\mathbb{X} \times \mathbb{L}}  {\pi}\big((\mathbf{x},l)\cup {\widetilde X}\big)\delta {\widetilde X} \\ 
  D(\mathbf{x}) & \triangleq \sum\limits_{l \in \mathbb{L}} {\widetilde D}(\mathbf{x},l) \label{eq:def-LRFS-phd-nolabel} 
\end{align}

In this paper, what indicated by the PHD is the unlabeled PHD by default unless otherwise stated.

\subsection{PHD of the Poisson, MB and LMB} \label{sec:Classic-RFS-distributions}
We hereafter omit the time notation and the dependence of the filter estimate on the observation process.
\subsubsection{Poisson RFS}

The {PHD of the Poisson RFS $X$ with mean $\lambda$ and MPD $f^{\text {p}}(X) = {{\text{e}}^{ - {\lambda}}}\prod_{\mathbf{x} \in X} {{\lambda}{s}(\mathbf{x})}$} is given by
\begin{align}
	D^{\text {p}}(\mathbf{x}) = \lambda {s}(\mathbf{x}) \label{eq:Poisson-PHD}
\end{align}
where $s(\mathbf{x})$ denotes the single-target probability density (SPD).

\subsubsection{MB}

An MB RFS $ X_n$ is the union of $n$ independent Bernoulli RFSs  \cite{Vo09CBmember} which can represent maximum $n$ targets. 
Denoting the $\ell$-th Bernoulli component (BC) by $\big(r^{(\ell)},s^{(\ell)}(\mathbf{x})\big)$, the {PHD of MB RFS $ X_n$ with the MPD $f^\text{mb} (X_n) = \sum_{\uplus_{\ell=1}^{n}X^{(\ell)}=X}\prod_{\ell=1}^{n}f^\text{b} \left(X^{(\ell)}\right)$ is} given by
\begin{align}
	D^{\text {mb}}(\mathbf{x}) = \sum_{\ell=1}^n{r^{(\ell)}}{s^{(\ell)}}(\mathbf{x})  \label{eq:PHD-MB}
\end{align}
where $\uplus$ denotes the disjoint union and the density of a Bernoulli RFS $X^{(\ell)}$ with target existence probability $r^{(\ell)}$ and SPD $s^{(\ell)}(\mathbf{x})$ is given by
{$f^{\text{b}}\left(X^{(\ell)}\right) = 1-r^{(\ell)}$ if $X^{(\ell)}  =\emptyset$, $f^{\text{b}}\left(X^{(\ell)}\right) = r^{(\ell)}s^{(\ell)}\left(\mathbf{x}\right) $ if $X^{(\ell)}  =\left\{ \mathbf{x}\right\}$ and $f^{\text{b}}\left(X^{(\ell)}\right) = 0 $ otherwise}.

\subsubsection{LMB}
Similar to the MB filter, the LMB filter \cite{Reuter14LMB} associates each labeled BC $l \in L \in \mathbb{L}$ with SPD ${s}(\mathbf{x},l) $ and existence probability $r^{(l)}$. {The MPD is much complicated as than that for the MB filter, which can be given as $\pi^{\text {lmb}}(\widetilde{X}) = \Delta \big(\widetilde{X}\big) \omega\big(\mathcal{L}\big(\widetilde{X}\big)\big) \prod_{(\mathbf{x},l) \in \widetilde{X}} {{s}(\mathbf{x},l)}$} 
where the distinct label indicator $\Delta \big(\widetilde{X}\big)= \delta_{|\widetilde{X}|}\big(|\mathcal{L}\big(\widetilde{X}\big)|\big)$, ${\delta_I}(L) = 1$ if $I=L$ and ${\delta_I}(L) = 0$ otherwise, the projection function $\mathcal{L}(\mathbb{X} \times \mathbb{L}) \rightarrow \mathbb{L}$ is given by $\mathcal{L}\big((\mathbf{x},l)\big) = l$, $\int_{\mathcal{X}} {{s}(\mathbf{x},{l})d\mathbf{x}}  = 1$, and $\omega(L)$ denotes the hypothesis weight corresponding to the label set $L \in \mathbb{L}$, which is given by $\omega(L) = \prod_{i\in \mathbb{L}} (1-r^{(i)}) \prod_{l\in L} \frac{{\mathrm{1}_\mathbb{L}}(l)  r^{(l)}}{1-r^{(l)}}$,
where ${\mathrm{1}_\mathbb{L}}(l) = 1$ if $l \in \mathbb{L}$ and ${\mathrm{1}_\mathbb{L}}(l) = 0$ otherwise and $\sum_{L \in \mathbb{L}} {\omega(L)} = 1$.
The labeled and unlabeled PHD of the LMB are respectively given as follows
\begin{align}
{\widetilde D}^{\text {lmb}}(\mathbf{x},l) &{= \sum\limits_{L \in \mathbb{L}} {\mathrm{1}_L}(l)  {\omega{(L)} {{s}(\mathbf{x},l)}} }\nonumber \\
& = r^{(l)} {s}(\mathbf{x},l) \\ 
{D}^{\text {lmb}}(\mathbf{x}) &= \sum\limits_{l \in L } {\widetilde D}^{\text {lmb}}(\mathbf{x},l) \nonumber \\
& = \sum\limits_{l \in L } r^{(l)} {s}(\mathbf{x},l) \label{eq:LMB-phd} 
\end{align}

%
\subsection{PHD-AA Consistency}
For a set of PHDs $D_i(\mathbf{x})$ produced by RFS filters $i \in \mathcal{I} =\{1,2,...,I\}$, the AA fusion is simply given as follows
\begin{equation}\label{eq:PHD-AA}
{D_{\text{AA}}}(\mathbf{x}) \triangleq \sum\limits_{i \in {\mathcal{I}}} {{w_i}{D_i}(\mathbf{x})}
\end{equation}
where the fusion weights $w_1,\dots,w_I >0, \sum_{i \in {\mathcal{I}}} {w_i}  =1$, and the weight space $\mathbb{W} \triangleq\{\mathbf{w} \in \mathbb{R}^{I}|\mathbf{w}^\mathrm{T}\mathbf{1}_I = 1, w_i > 0, \forall i \in \mathcal{I}\}$.


The AA is a Fr\'{e}chet mean in the sense of the integrated squared difference (ISD) \cite{Li20AAmb}, i.e.,
\begin{equation}
  D_\mathrm{AA}(\mathbf{x})  = \operatorname*{arg\,min}_{g\in\mathcal{F}_{\mathcal{X}}} \sum\limits_{i \in {\mathcal{I}}}{w_i\int_{\mathcal{X}} \big(D_i(\mathbf{x})-g(\mathbf{x})\big)^2 \delta \mathbf{x}}   \label{eq:FrechetAAISD}
\end{equation}
where $\mathcal{F}_\mathcal{X} \triangleq \{f: \mathcal{X} \rightarrow \mathbb{R} \}$ 

Relevantly, it is more known as follows, which we refer to as the best fit of the mixture (BFoM) \cite{Li23AApmbm}, 
\begin{equation}\label{eq:RFS-AA-Whole-KLD}
 {D_{\text{AA}}}(\mathbf{x}) = \operatorname*{arg\,min}_{g\in\mathcal{F}_\mathcal{X}} \sum\limits_{i \in {\mathcal{I}}} {w_iD_\text{KL}\big(D_i\|g\big)}
\end{equation}
where the Kullback-Leibler (KL) divergence, {extended to the PHD domain,} is defined as $D_\text{KL}\big(f\|g \big) \triangleq \int_\mathcal{X} {f(\mathbf{x})\log \frac{f(\mathbf{x})}{g(\mathbf{x})} \delta \mathbf{x}}$.

However, in the context of labeled MPD fusion, the labels which suffer more or less from estimate errors conditional on random false/missing data and noises are often different across the sensors. They need to be pairwise-matched between the labeled filters {to comply with the} track/label-wise fusion rule \cite{Li22RFS-AA-Derivation}. Nevertheless, a proper quantitative
definition of the divergence/distance between labels is still missing, which is actually the base for label matching. 
Different labeling matching choices will {in turn} lead to very different divergence values and correspondingly different fusion results \cite{Li23BMsurvey}.
To overcome this challenge, the fusion in this work is limited in the unlabeled domain, {more precisely unlabeled PHD-AA fusion,} which is free of label matching, although the LRFS/LMB filters are involved.

\section{GM Representations of MB/LMB-PHD} \label{GM-implementation} 

\subsection{GM-PHD} \label{sec:GM-PHD}
It is theoretically justified and also practically convenient to represent the RFS posterior and at the same time the corresponding PHD by a GM, which facilitates the GM-PHD-AA fusion. Actually, what is estimated and propagated over time in the PHD filter \cite{Mahler03} is the PHD, not the MPD. 
Straightforwardly, the GM approximation of the PHD filter $i \in \mathcal{I}$ at filtering time $k$ can be written as \cite{Vo06}:
\begin{equation}\label{eq:GM_PHD}
D_{i, k}(\mathbf{x}) \approx \sum_{j=1}^{J_{i,k}} \omega_{i,k}^{(j)} \mathcal{N}(\mathbf{x};\bm{\mu}_{i,k}^{(j)},\bm{\Sigma}_{i,k}^{(j)}) 
\end{equation}
where $\mathcal{N}\big(\mathbf{x};\bm{\mu}\rmv,\bm{\Sigma}\big)$ denotes a Gaussian PDF with mean vector $\bm{\mu} $ and covariance matrix $\bm{\Sigma}$, $J_{i,k}$ is the number of GCs in total, and $\omega_{i,k}^{(j)}$ is the weight of $j$th GC at sensor $i$.

The whole PHD is thereby completely determined by the parameter set $\mathcal{G}_{i,k} \triangleq \big\{ \big( \omega_{i,k}^{(j)} , \bm{\mu}_{i,k}^{(j)} , \bm{\Sigma}_{i,k}^{(j)}  \big) \big\}_{j=1,\dots,J_{i,k}}$. The expected number of targets at sensor $i\in \mathcal{I}$ at time $k$ can be approximated by
{$\hat{N}_{i,k} = \sum_{j=1}^{J_k} \omega_{i,k}^{(j)}$}.

\subsection{GM Representation of the MB PHD} \label{sec:MB-AA}
Consider the GM implementation of the MB posterior 
represented by a set $L_{i,k}$ of BCs at filtering time $k$ by sensor $i$. Each BC $\big(r_{i,k}^{(\ell)}, s_{i,k}^{(\ell)}(\cdot)\big), \ell \in L_{i,k}$ is represented by $J_{i,k}^{(\ell)}$ GCs weighted by $\omega_{i,k}^{(\ell,\iota)} \!\ge\rmv 0$, $\iota=1,\dots,J_{i,k}^{(\ell)}$, i.e.,
\begin{equation}
s_{i,k}^{(\ell)}(\mathbf{x}) = \sum _{\iota=1}^{J_{i,k}^{(\ell)}} \omega_{i,k}^{(\ell,\iota)} \ist \mathcal{N}\big(\mathbf{x};\bm{\mu}_{i,k}^{(\ell,\iota)}\rmv,\bm{\Sigma}_{i,k}^{(\ell,\iota)}\big) \label{eq:BC_GM} 
\end{equation}
where
\begin{equation}\label{eq:BCweightofGC}
  \sum _{\iota=1}^{J_{i,k}^{(\ell)}} \omega_{i,k}^{(\ell,\iota)} = 1
\end{equation}

Each BC is completely determined by the {GM} parameter set $\mathcal{G}^{(\ell)}_{i,k} \triangleq \big\{ \big( \omega_{i,k}^{(\ell,\iota)} , \bm{\mu}_{i,k}^{(\ell,\iota)} , \bm{\Sigma}_{i,k}^{(\ell,\iota)}  \big) \big\}_{\iota=1,\dots,J_{i,k}^{(\ell)}}$ {by which the PHD \eqref{eq:PHD-MB}} is approximated as a set of GMs as follows
\begin{equation}\label{eq:mbPHD_GM}
{D_{i,k}}(\mathbf{x}) \approx \sum_{\ell \in L_{i,k}} r_{i,k}^{(\ell)} \sum _{\iota=1}^{J_{i,k}^{(\ell)}} \omega_{i,k}^{(\ell,\iota)} \mathcal{N}\big(\mathbf{x};\bm{\mu}_{i,k}^{(\ell,\iota)}\rmv,\bm{\Sigma}_{i,k}^{(\ell,\iota)}\big)
\end{equation}
which may be expressed as a unified GM of the parameter set $\mathcal{G}_{i,k} \triangleq \big\{ \mathcal{G}^{(\ell)}_{i,k} \big\}_{\ell \in L_{i,k}}$ as follows, c.f., \eqref{eq:GM_PHD},
\begin{equation}\label{eq:mbPHD_UnifiedGM}
{D_{i,k}}(\mathbf{x}) \approx \sum_{j=1}^{J_{i,k}} \omega_{i,k}^{(j)} \mathcal{N}(\mathbf{x};\bm{\mu}_k^{(j)},\bm{\Sigma}_k^{(j)})
\end{equation}
where $J_{i,k} = \sum_{\ell \in L_{i,k}}  J_{i,k}^{(\ell)}$ and the recombined weight of each GC labeled as $j$ is given as
\begin{equation}\label{eq:mbPHD_WeightGC}
\omega_{i,k}^{(j)} = r_{i,k}^{(\ell)}  \omega_{i,k}^{(\ell,\iota)}
\end{equation}
which uses the following unique index mapping at each sensor $i\in \mathcal{I}$ at time $k$,
\begin{equation}\label{eq:indexMapping}
  j \leftrightarrow (\ell,\iota)
\end{equation}
where $j=1,...,J_{i,k}, \ell  \in L_{i,k}, \iota=1,...,J^{(\ell)}_{i,k}$.

The mean-state of the $\ell$-th BC 
as expressed in \eqref{eq:BC_GM} is calculated by
$\bar{\bm{\mu}}_{i,k}^{(\ell)} = \sum _{\iota=1}^{J_{i,k}^{(\ell)}} \omega_{i,k}^{(\ell,\iota)} \bm{\mu}_{i,k}^{(\ell,\iota)} $. 
By integrating the PHD, 
the number of targets is estimated at sensor $i$ by
\begin{align}\label{eq:MB_Card}
  \hat{N}_{i,k}  & = \sum_{\ell\in L_{i,k}} r_{i,k}^{(\ell)} \nonumber \\
  & =  \sum _{j=1}^{J_{i,k} } \omega_{i,k}^{(j)}
\end{align}

\subsection{GM Representation of the LMB PHD} \label{sec:MB-PHD}
The LMB can be viewed as a special MB with assigned label for each BC. That is, with slight abuse of notation and by interpreting each component index $\ell$ in \eqref{eq:BC_GM}, \eqref{eq:mbPHD_GM} and \eqref{eq:MB_Card} as a track label $l$ and set $L_{i,k}$ as the label set of sensor $i$ at time $k$, the above GM formulation \eqref{eq:mbPHD_GM} for the MB-PHD can be the same derived for calculating the unlabeled PHD of the LMB. 
Note that the labels are usually ordered pairs of integers {$l = (k, \kappa)$}, where $k$ is the time of birth, and $\kappa$ is a unique index to distinguish new targets born at the same time \cite{Vo13Label}. 
The unlabeled PHD $D_{i,k}(\mathbf{x})$ of the local LMB at time $k$ at sensor $i \in \mathcal{I}$, expressed by the parameter set $\mathcal{G}_{i,k} \triangleq \big\{ \big( \omega_{i,k}^{(j)} , \bm{\mu}_{i,k}^{(j)} , \bm{\Sigma}_{i,k}^{(j)}  \big) \big\}_{j=1,\dots,J_{i,k} }$, can be given by, c.f., \eqref{eq:mbPHD_UnifiedGM},
\begin{equation}
{D_{i,k}}(\mathbf{x}) = \sum_{j=1}^{J_{i,k}} \omega_{i,k}^{(j)} \mathcal{N}(\mathbf{x};\bm{\mu}_k^{(j)},\bm{\Sigma}_k^{(j)})
\end{equation}
where {multiple L-GCs may have the same label}.

\section{GM-PHD Fit via Optimizing L-GC Weights } \label{sec:GMfit}
This paper addresses averaging the PHDs of the PHD, MB, and LMB filters via the GM implementations, where the GM parameter set for time $k$ at sensor $i \in \mathcal{I}$ is denoted by $\mathcal{G}_{i,k}$ and the corresponding local PHD by $D_{i,k}(\mathbf{x})$. {The goal of the PHD-AA fusion is to update the local GM parameters so that their corresponding PHDs reach consensus as calculated by \eqref{eq:PHD-AA} in each filtering iteration. As a result, the local filters have similar or the same PHD while their MPDs are variant. A simple 1-D example is given in Fig. \ref{fig:PHDvsMB}, where the MB filter has the same PHD as that of the PHD filter while their filter posteriors are heterogeneously parameterized.}

\begin{figure}
  \centering
  {
  \includegraphics[width=8.5cm]{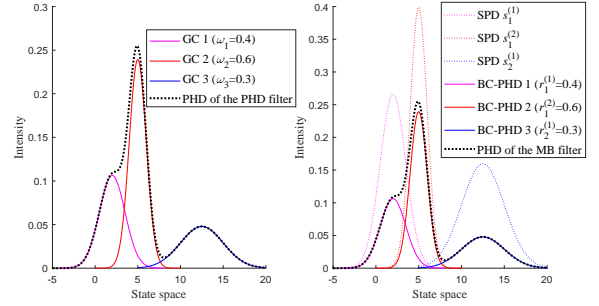}\\  
  \caption{A MB filter reaches PHD consensus with a PHD filter 
   } \label{fig:PHDvsMB}
  \vspace{-4mm}
  }
\end{figure}

\begin{remark}
The key challenge to heterogeneous RFS filter fusion origins from the fact that different RFS filters result in different types of MPDs. 
Unlike the GM-PHD filter, the GCs in the GMs of the MB and of LMB filters intrinsically belong to different (unlabeled or labeled) BCs. Therefore, if a new L-GC is created or an existing L-GC is deleted, one has to identify the BC/label it belongs to.
\end{remark}

The situation is more complicated in the case of MB mixture fusion \cite{Li23AApmbm} 
where the measurement-to-track association hypothesis is further involved.
To address this challenge, we hereafter propose an approximate GM-PHD fitting approach which does not create or disregard any L-GCs in the local filters but only optimize their weights $\omega_{i,k}^{(j)}, j=1,...,J_{i,k}$ for consensus over their PHDs. Clearly, the resulted PHD is no more than an approximate to the {exact} PHD-AA. 
This fit does not only save computation but also enables the parallel-communication-filtering operation similarly as was done in \cite{Li17PCsmc,Li19ParallelCC} since the mean $\bm{\mu}_{i,k}^{(j)}$ and covariance $\bm{\Sigma}_{i,k}^{(j)}$ of each L-GC $j$, as well as the total number $J_{i,k}$ of L-GCs, are unchanged and therefore can be updated in parallel during the fusion. We omit the detail here. Obviously, better fit can be expected if $\bm{\mu}_{i,k}^{(j)}$ and $\bm{\Sigma}_{i,k}^{(j)}$ are also optimized jointly with the weights, regardless of the computational cost.

\subsection{Minimizing ISD of GMs via Reweighting} \label{sec:ise}
To evaluate the {goodness} of fit, the KL divergence can be considered which, however, does not admit analytical solution for the GMs. For the sake of computational efficiency, one may consider the Cauchy-Schwarz divergence \cite{Kampa11CSD} \footnote{{It is necessary to note that the Cauchy-Schwarz divergence was originally defined on PDFs and is a normalized distance. Therefore, its application to the PHD is not straightforward but proper extension is needed.}} 
and the ISD metric as given in \eqref{eq:FrechetAAISD}, both of which allow analytical calculation for GMs. We consider the latter only, which complies with the BFoM as expressed by \eqref{eq:FrechetAAISD}. 
Formally speaking, the goal of our approach is to determine the new weight for each L-GC at the local sensor in the following BFoM sense
\begin{align}
  \{ \omega_{i,k}\}^{\text{BFoM}} & = \operatorname*{arg\,min}_{\{ \omega_{i,k}\} \geq 0} \text{ISD}\big(D_{i,k}\| D_{\text{AA},k}\big) \label{eq:opt_W_ISD} \\
  & = \operatorname*{arg\,min}_{\{ \omega_{i,k}\} \geq 0} \int \bigg( (1-w_i) \sum_{j=1}^{J_{i,k}} \omega_{i,k}^{(j)} \mathcal{N}(\mathbf{x};\bm{\mu}_{i,k}^{(j)},\bm{\Sigma}_{i,k}^{(j)}) \nonumber \\ 
  & \hspace{3mm} - \sum_{s\in \mathcal{I}\setminus{i}} w_s \sum_{j=1}^{J_{s,k}} \omega_{s,k}^{(j)} \mathcal{N}(\mathbf{x};\bm{\mu}_{s,k}^{(j)},\bm{\Sigma}_{s,k}^{(j)}) \bigg)^2 d \mathbf{x} \label{eq:WeightISDfit}  
\end{align}
where $\big\{ \omega_{i,k}^{(j)}\big\}_{j=1,...,J_{i,k}}$ was written as $\{ \omega_{i,k}\}$ in short and $A\setminus B$ is the set difference of $A$ and $B$. 


{In addition to the above PHD-GM fit, we further communicate and average the multitarget set cardinality estimates of local sensors for more accurate estimation of the target number}, namely cardinality consensus (CC) \cite{Li19CC}, although the cardinality (distribution) is not directly estimated in the PHD, MB and LMB filters. By using \eqref{eq:MB_Card}, the AA of the estimated number of targets can be calculated as
{\begin{align}\label{eq:AA-card}
\hat{N}_{\text{AA},k} & = \sum_{i \in \mathcal{I}}  \hat{N}_{i,k} 
\end{align}
where $\hat{N}_{i,k} = \sum_{j=1}^{J_k} \omega_{i,k}^{(j)}$ as addressed earlier for the PHD filter in Sec. \ref{sec:GM-PHD} and for the MB/LMB filters in Sec. \ref{sec:MB-PHD}.}

{By taking into account both the non-negative constraint and the above CC constraint of the weights, the GM-weight-fit problem can be formulated as the following constrained multivariate optimization (CMO) problem
  \begin{align}
    &\operatorname*{min} J\big(\{ \omega_{i,k}\}\big) = \text{ISD}\big(D_{i,k}\| D_{\text{AA},k}\big) \label{eq:CMO} \\
    &\text{s.t.} \hspace{3mm} \omega_{i,k}^{(j)} \geq 0, j=1,...,J_{i,k} \\
    & \hspace{6mm} \sum_{j=1}^{J_k} \omega_{i,k}^{(j)} = \sum_{j'=1}^{J'_k} \omega_{i',k}^{(j')}, \forall i' \neq i \label{eq:CC_constraint}
  \end{align}}

{The most challenging issue of the above CMO problem is from the CC constraint \eqref{eq:CC_constraint}. To simplify this problem, we address it in a separate step at the end of our algorithm; see section \ref{sec:CC}. By removing it temporally, the left inequality CMO problem leads to the following Lagrangian function
\begin{align}
  L\big( \{ \omega_{i,k}\}, \{\lambda_j\} \big)= J\big(\{ \omega_{i,k}\}\big) - \sum_{j=1}^{J_{i,k}} \lambda_j \omega_{i,k}^{(j)} 
\end{align}
where $\{\lambda_j\}_{j=1,...,J_{i,k}}$ are the Lagrange multipliers. }

{The Karush-Kuhn-Tucker conditions for this CMO problem are given as follows
\begin{align}\label{eq:KKT}
  \frac{\partial J\big(\{ \omega_{i,k}\}\big) }{\partial \omega_{i,k}^{(j)}} - \lambda_j & = 0, j=1,2,..., J_{i,k} \\
  \lambda_j \omega_{i,k}^{(j)} & = 0, j=1,2,..., J_{i,k}  \\
  \lambda_j & \geq 0, j=1,2,..., J_{i,k}
\end{align}}


\subsection{A Coordinate Descent CMO Solver and Over-fit}

\begin{lemma}
{Relaxing the non-negative constraint for the weights $\omega_{i,k}^{(j)}\geq 0$,} \eqref{eq:WeightISDfit} can be solved by, $\forall i \in \mathcal{I}, j=1,...,J_{i,k}$,
\begin{align}
 & \omega_{i,k}^{(j),\text{BFoM}} = \nonumber \\
 & \frac{\Delta} {(1-w_i)} {\sum\limits_{s\in \mathcal{I}\setminus{i}} \sum\limits_{l=1}^{J_{s,k}}{w_{s} \omega_{s,k}^{(l)}}\mathcal{N}\big({\bm{\mu}_{i,k}^{(j)}};{\bm{\mu}_{s,k}^{(l)}},{\bm{\Sigma}_{i,k}^{(j)}}{\rm{ + }}{\bm{\Sigma}_{s,k}^{(l)}}\big)}\nonumber \\
 &- {\Delta} \sum\limits_{j'\in \mathcal{J}_i^{-j} } {\omega_{i,k}^{(j')}\mathcal{N}\big({\bm{\mu}_{i,k}^{(j)}};\bm{\mu}_{i,k}^{(j')},{\bm{\Sigma}_{i,k}^{(j)}} + {\bm{\Sigma}_{i,k}^{(j')}}\big)} \label{eq:minISD-wi}
\end{align}
where $\Delta = {\left| {2 \bm{\Sigma}_{i,k}^{(j)} } \right|}^{1/2} { (2\pi )}^{d/2} $ and $\mathcal{J}_i^{-j} = \{1,..,J_{i,k}\} \setminus j$.
\end{lemma}

\begin{proof}
See Appendix \ref{appendix:iseMIN}. The result is from \eqref{eq:minISD-a}.
\end{proof}

By using \eqref{eq:minISD-wi}, \revisNew{we propose a coordinate descent method
that sequentially updates} the weight of each L-GC $j=1,...,J_{i,k}$ for one or multiple iterations. 
That is, in iteration $t\in \mathbb{N}^+$, $\omega_{i,k}^{(j')}, j'\in \mathcal{J}_i^{-j} $ in the right hand of \eqref{eq:minISD-wi} is defined as follows for calculating $\omega_{i,k}^{(j,t),\text{BFoM}}$,
\begin{equation}\label{eq:omega_iteration}
  \omega_{i,k}^{(j')} :=  \begin{cases}
\omega_{i,k}^{(j',t)} &  j'< j\\
\omega_{i,k}^{(j',t-1)}  & j'>j
\end{cases}
\end{equation}
where $\omega_{i,k}^{(j,t)}$ denotes the updated weight of the $j$th L-GC in iteration $t$ and $\omega_{i,k}^{(j,0)} :=\omega_{i,k}^{(j)}$.

The above \revisNew{coordinate descent} approach to the {CMO} problem {often suffers} from over-fit, {i.e., the non-negative constraint $\{ \omega_{i,k}\} \geq 0$ is violated or the resulting weight is over-large}.
This can be addressed as follows, $\forall i \in \mathcal{I}, j=1,...,J_{i,k}$,
\begin{align}
 \tilde{\omega}_{i,k}^{(j,t)} & = \text{max}\Big( \epsilon,  \omega_{i,k}^{(j,t),\text{BFoM}} \Big) \label{eq:bound-w}
\end{align}
where $\epsilon \geq 0$ is a small constant (e.g., 0.01) specified to avoid generating negative weight. 

More importantly, we further take the following strategy to avoid over-fit
\begin{align}
 \omega_{i,k}^{(j,t)} & = \alpha_t \tilde{\omega}_{i,k}^{(j,t)} + (1-\alpha_t) \omega_{i,k}^{(j,t-1)} \label{eq:upd-w-alpha}
\end{align}
where $0<\alpha_t<1$ is a scaling factor set for iteration $t$
, which can be interpreted as a learning rate.

{The learning rate $\alpha_t$ may be set simply fixed. 
However, there is no guarantee for the sequential fit approach with fixed learning rate $\alpha$ to converge to the optimal solution except that $\alpha$ gradually reduces in a proper means with the increase of $t$. In the latter, we suggest the following strategy by employing a hyper-parameter $\beta_t$
\begin{equation}\label{eq:LearningRateRule}
  \alpha_{t+1} = \beta_t\alpha_t
\end{equation}
where $0 < \beta_t \leq 1$ can be interpreted as a fading rate and $\beta_t =1$ means a fixed learning rate.}

The above {GM-weight-fit iteration based on either fixed or adaptive learning rate} may 
be terminated at a convergence level (i.e., local PHDs approximately reach consensus, e.g., $\text{ISD}\big(D_{i,k}\| D_{\text{AA},k}\big)\leq \varepsilon$ where $\varepsilon$ is a small threshold) or up to a maximum number of fit iterations, e.g., $t\leq t_{\text{max}}= 3$. 

\subsection{Fusion Feedback} \label{sec:fusionfeedback}
\subsubsection{PHD filter} Each L-GC $j=1,...,J_{i,k}$ of sensor $i\in \mathcal{I}$ that is used for representing the PHD will be simply re-weighted as $\omega_{i,k}^{(j,t)}$ as calculated in \eqref{eq:upd-w-alpha} after $t$ GM-weight-fit iterations at time $k$, which can be written as follows
\begin{align}
\omega_{i,k}^{(j),\text{upd}} = \omega_{i,k}^{(j,t)}
\end{align}
\subsubsection{MB/LMB filter}
Based on the unique mapping \eqref{eq:indexMapping}, we get the updated, un-normalized weight for $\iota$th L-GC of the $\ell$th BC/track in sensor $i$ at time $k$ from the fused weight $\omega_{i,k}^{(j,t)}$ as calculated in \eqref{eq:upd-w-alpha} after $t$ GM-weight-fit iterations as follows
\begin{align}
\tilde{\omega}_{i,k}^{(\ell,\iota),\text{upd}} = \omega_{i,k}^{(j,t)}
\end{align}
where $j=1,...,J_{i,k}, \ell \in L_{i,k}, \iota=1,...,J^{(\ell)}_{i,k}$.

According to \eqref{eq:BCweightofGC}, the weight of the L-GCs in each BC/track should then be normalized, i.e.,
\begin{equation}
\omega_{i,k}^{(\ell,\iota),\text{upd}} = \frac{\tilde{\omega}_{i,k}^{(\ell,\iota),\text{upd}}}{\sum _{\iota=1}^{J_{i,k}^{(\ell)}} \tilde{\omega}_{i,k}^{(\ell,\iota),\text{upd}}}
\end{equation}

%
%
\subsubsection{CC} \label{sec:CC}
 Regardless that the cardinality is not really used for estimating the number of targets in the local MB/LMB filters \footnote{As usually done in the standard single sensor filter, the number of targets is estimated from the posterior cardinality distribution by taking its mode in the MB filter \cite[Sec. III.D]{Vo09CBmember} while it depends on the existing probability of each track in comparison with application specific thresholds \cite[Sec. IV.E]{Reuter14LMB}.}, it can be used to re-scale the L-GC weights $\omega_{i,k}^{(j)}, j=1,2,...,J_{i,k}$ in the PHD filter and re-scale the track existence probability $r_{i,k}^{(\ell)}, \ell\in L_{i,k}$ in the case of MB/LMB filter, respectively, as follows
\begin{align}
\omega_{i,k}^{(j),\text{upd}} &=  \frac{\omega_{i,k}^{(j,t)} \hat{N}_{\text{AA},k}}{\hat{N}_{i,k}} \label{eq:CC-PHD} \\
  r_{i,k}^{(\ell),\text{upd}} &=  \frac{r_{i,k}^{(\ell)} \hat{N}_{\text{AA},k}}{\hat{N}_{i,k}} \label{eq:CC-MB}
\end{align}

\subsection{Insufficiency for Track Fusion}\label{sec:insufficiencyofB2B}
The proposed unlabeled PHD-AA suits the most the PHD filter but not the MB/LMB filters due to two reasons 
\begin{remark} \label{remark:insufficientPHD1}
First, the PHD is the first order moment of the MPD for which the PHD-consensus is insufficient for MB/LMB consensus. In other words, two MB/LMB densities can be variant even if their PHDs are the same. This is the limitation of existing RFS-AA fusion approaches {based on the consensus of the PHD not of the MPD} \cite{Li22RFS-AA-Derivation}.
\end{remark}
\begin{remark}\label{remark:insufficientPHD_T2T}
Second, the fusion of the MBs/LMBs needs to properly match/associate the BCs {that represent latently the same target} 
and then carry out the fusion in each associated group of BCs; 
see BC-to-BC (B2B) association methods \cite{Li20AAmb,Yi21Heterogeneous} 
in the homogeneous fusion case. While the B2B-based AA fusion will result in an exactly averaged BC \cite{DaKai_Li_DCAI19,Li19Bernoulli}, it is, however, inapplicable when fusing the MB/LMB filters with the PHD filters that do not have distinct BCs.
\end{remark}

To say the least, the label matching is nontrivial as the labels can be very different with each other among diverse LMB filters. This can be illustrated in Fig. \ref{fig:PHDlabel} where four GM-LMBs with completely the same GM parameters $\mathcal{G}_i = \mathcal{G}_j$ but different labels $L_i \neq L_j$, $\forall i \neq j$. Different labeling corresponds to different labeled PHDs and different target-state estimates, although their unlabeled PHDs are identical.
Overall, our {present heterogeneous GM-weight-fit approach} overcomes these challenges as it does neither seek track-to-track fusion nor fuse in the labeled domain. 

\begin{figure}
  \centering
  \includegraphics[width=8.5cm]{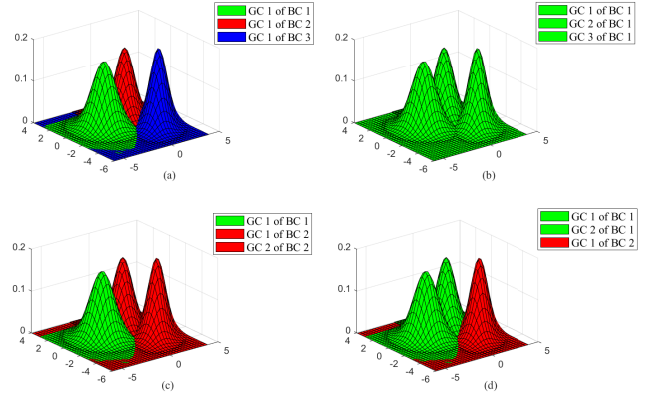}\\  
  \caption{Four LMBs of the same unlabeled PHD but different labeled PHDs and different potential target-state-estimates. 
   } \label{fig:PHDlabel}
  \vspace{-4mm}
\end{figure}

\section{Simulations} \label{sec:simulation}
We consider a ROI given by $[-1\ist\text{km}, 1\ist\text{km}]\times [-1\ist\text{km}, 1\ist\text{km}]$ which is viewed by 4 sensors. 
The target state is denoted as $\mathbf{x}_k=[x_k \; \dot{x}_k \; y_k \; \dot{y}_k]^\text{T}\rmv$ with planar position $[x_k \; y_k]^\text{T}$ and velocity $[\dot{x}_k \; \dot{y}_k]^\text{T}$.
There are totally 12 targets born at different points as shown in Fig. \ref{fig:scenario}.
The target birth is modeled by a MB process with parameters $\{r_{\text{B}} ,p_{\text{B}}^{(\ell)}(\cdot) \}_{\ell=1}^{4}$ at each time-instant, where $r_{\text{B}} =0.03$%
, and $p_{\text{B}}^{(\ell)}(\mathbf{x})=\mathcal{N}(\mathbf{x};\mathbf{\mu}_\text{B}^{(\ell)},\bm{\Sigma}_\text{B})$ with parameters $ \mathbf{\mu}_\text{B}^{(1)}=[0,0,0,0]^\text{T}$, $\mathbf{\mu}_\text{B}^{(2)}=[400\text{m},0,-600\text{m},0]^\text{T}$, $\mathbf{\mu}_\text{B}^{(3)}=[-800\text{m},0,-200\text{m},0]^\text{T}$, $\mathbf{\mu}_\text{B}^{(4)}=[-200\text{m},0,800\text{m},0]^\text{T}$, $
\bm{\Sigma}_\text{B}^{\text{ \ \ \ }}=\mathrm{diag}([10\text{m},10\text{m/s},10\text{m},10\text{m/s}]^\text{T})^{2}$. {We note that this new-born target MB modeling matches the MB/LMB filters but not the PHD filter. 
In practice, however, it is never known what the ground truth is. The local sensors assume the target birth model based on their own prior knowledge or needs. Both Poisson and MB models are reasonable. }

Each target has a constant survival probability $0.95$ and follows a constant velocity motion (with the sampling interval $\Delta =1$s) using noiseless transition density $f_{k|k-1}(\mathbf{x}_{k}|\mathbf{x}_{k-1})=\mathcal{N}(\mathbf{x}_{k};F \mathbf{x}_{k},\mathbf{0})$ for generating the ground truth while the filters use $f_{k|k-1}(\mathbf{x}_{k}|\mathbf{x}_{k-1})=\mathcal{N}(\mathbf{x}_{k};F \mathbf{x}_{k},\mathbf{Q})$, where 
$$
F=\mathbf{I}_{2}\otimes\left[\begin{array}{cc}
1 & \Delta\\
0 & 1
\end{array}\right], \,\mathbf{Q}=25 \times\mathbf{I}_{2}\otimes\left[\begin{array}{cc}
\Delta^{2}/2 & \Delta/2\\
\Delta/2 & \Delta
\end{array}\right] 
$$
where $\otimes$ is the Kronecker product.

The simulation is performed 100 runs with conditionally independent measurement series for 100 seconds each run. {In the following, we first} consider four sensors with the same time-invariant target detection probability $p_d = 0.9$ and linear measurement model of senor $s$ as follows
\begin{equation}\label{eq:observationModel}
\mathbf{z}_{s,k}= \left[ \begin{array}{cccc}
1 & 0 & 0 & 0 \\
0 & 0 & 1 & 0 \\
\end{array} \right] \mathbf{x}_k+ \left[ \begin{array}{c}
v_{1,k} \\
v_{2,k} \\
\end{array} \right]  
\end{equation}
with $v_{1,k}$ and $v_{2,k}$ as mutually independent zero-mean Gaussian noise with the same standard deviation of $10$m.

{We then consider four sensors with the same $p_d = 0.9$ but nonlinear measurement model as follows
\begin{equation} \label{eq:RangeBearing}
\mathbf{z}_{s,k}= \begin{bmatrix}
\sqrt{( x_k \!-\! x^{(s)} )^2 + (y_k \!-\! y^{(s)})^2} \,\,\\
\vspace{0.5mm}
\tan^{-1}\!\Big( \frac{ x_k - x^{(s)} }{ y_k - y^{(s)} } \Big)
\end{bmatrix}
+ \begin{bmatrix} v_{s,k}^{(1)} \\[0.5mm] v_{s,k}^{(2)} \end{bmatrix}
\end{equation}
where $x^{(s)}$ and $y^{(s)}$ are the coordinates of sensor $s$ which are $[-500\text{m},-800\text{m}]^\text{T}$, $[-500\text{m},800\text{m}]^\text{T}$, $[600\text{m},800\text{m}]^\text{T}$ and $[600\text{m},-800\text{m}]^\text{T}$, respectively in our case, $v_{s,k}^{(1)}$ and $v_{s,k}^{(2)}$ are independent zero-mean Gaussian with standard deviation $\sigma_1 \!=\! 10\ist$m and $\sigma_2 \rmv=\rmv (\pi/90)\ist$rad, respectively. The field of view of each sensor is a disk of radius $2$km centered at the sensor position $[x^{(s)},y^{(s)}]^\text{T}$ which covers the entire ROI.}

{In both cases, the clutter measurements of different sensors are independent which uniformly distributed over each sensor's field of view for which the number of clutter points at each scan is Poisson with rate $\kappa_s=10$.
}

%

The filter performance is evaluated by the optimal subpattern assignment (OSPA) error \cite{Schuhmacher08}, which is given as follows,
$d^{(c,p)}_\text{ospa}(\hat{{X}}, {X}) = \Big( \big(|\hat{{X}}|\big)^{-1}\big( d_\text{Loc}(\hat{{X}}, {X})   +  d_\text{Card} (\hat{{X}}, {X}) \big) \Big)^{\frac{1}{p}}$ for $|\hat{{X}}| \geq |{X}|$,
where the localization error (OSPA Loc) and cardinality error (OSPA Card) are defined as
$ d_\text{Loc}(\hat{{X}}, {X})  = {\mathop {\min }\limits_{\pi  \in {{\rm \Pi} _{|\hat{{X}}|}}} \sum\limits_{i = 1}^{|{X}|} {{d^{(c)}}{{({{X}_i},{\mathbf{\hat{x}}_{\pi (i)}})}^p}} }$ and $ d_\text{Card} (\hat{{X}}, {X})  = { {{c^p}}  (|\hat{{X}}| - |{X}|)}$, respectively.
Here, $\pi$ and $ {\rm \Pi}_n $ are a permutation and the set of all permutations on $\{1,\ldots,n \}$, and $ {d^{(c)}}({X},\mathbf{y}) = \min \left( {d({X},\mathbf{y}),c} \right) $ is a metric between ${X}$ and $\mathbf{y}$ cut-off at $c$. If $|\hat{{X}}| < |{X}|$, $d^{(c,p)}_\text{ospa}(\hat{{X}}, {X}) = d^{(c,p)}_\text{ospa}({X}, \hat{{X}})$, $d_\text{Loc}(\hat{{X}}, {X}) = d_\text{Loc}({X},\hat{{X}})$, $d_\text{Card} (\hat{{X}}, {X})  = d_\text{Card} ({X}, \hat{{X}}) $. In our simulations, we use $c=100$m and $p=2$.


We consider both homogeneous and heterogeneous fusion cases. In the former, all four sensors run the same PHD, MB or LMB filters, respectively, while in the latter the sensors run different filters. Different levels of fusion have been considered in both cases. They may not cooperate with each other at all (namely noncooperative), only communicate and fuse with each other the estimated number of targets via \eqref{eq:CC-PHD} using $\omega_{i,k}^{(j,t)}=\omega_{i,k}^{(j)}$ or via \eqref{eq:CC-MB} (namely CC only), or perform the proposed GM-PHD-AA fusion in various numbers of GM-weight-fit iterations, $t=1,2,...,6$. 
{We use the learning rate $\alpha_1 = 0.2$ or $\alpha_1 = 0.4$  in \eqref{eq:upd-w-alpha} and $\beta_t=0$ (namely fixed learning rate), $\beta_t=0.6$ (namely adaptive learning rate) in \eqref{eq:LearningRateRule}, respectively.}
Furthermore, we use uniform fusing weights, no matter what filters/sensors are involved.

To set up the local filters, the maximum number of L-GCs in the local GM is $200$ for the PHD filters, the maximum number of tracks/BCs is $50$ and the maximum number of L-GC for each track/BC is $20$ in the MB/LMB filters. We note that in practice, these parameters should be designed according to the computational capacity and sensing rate of the local sensors. 
The other setup of the filters we used, as well as the ground truth, are the same as the standard one as given in the codes released by Vo-Vo at https://ba-tuong.vo-au.com/codes.html. {Codes for our simulations are available soon in the following URL: sites.google.com/site/tianchengli85/matlab-codes/aa-fusion.}

\begin{figure}
  \centering
  \includegraphics[width=8cm]{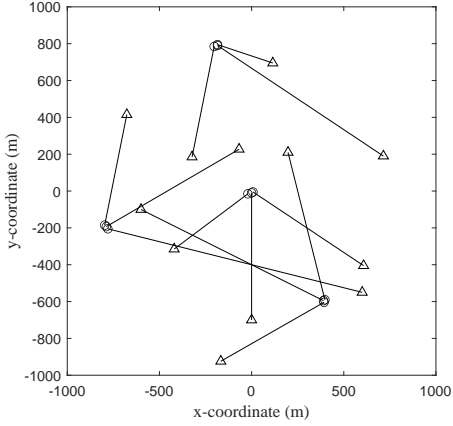}\\
  \caption{ROI and target trajectories starting from $\circ$ and ending at $\triangle$} \label{fig:scenario}
  \vspace{-2mm}
\end{figure}

\begin{figure}
  \centering
  \includegraphics[width=8cm]{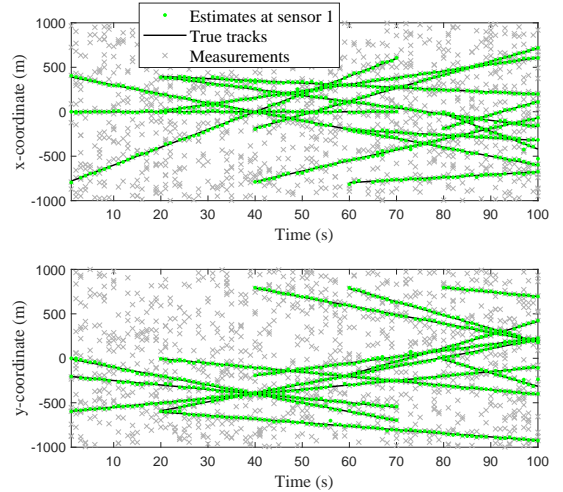}\\
  \caption{True tracks, measurements and target position estimates in one run by a local GM-PHD filter}
  \vspace{-2mm}
\end{figure}

\subsection{Homogeneous PHD/MB/LMB fusion} \label{sec:homogeneousSim}
We first consider the homogeneous PHD filter fusion in order to show the accuracy of the proposed GM fit approach based on the linear measurement model. {We also test the standard GM-PHD-AA fusion \cite{Li17PC} which averages all of the local GMs. Differently, the size of the local GM is constant in our proposed GM-PHD-fit approach while it increases and so requires appreciate GM reduction such as merging and pruning in the exact GM-PHD-AA fusion.}
The average OSPA errors, cardinality errors and localization errors over 100 runs for each filtering iteration at four sensors are shown in Fig. \ref{fig:PHD-opsa-I=1}. The average OSPA errors of the four PHD filters over all 100 runs of 100 filtering steps, {by using our proposed GM-weight-fit approach for PHD consensus based on either fixed or adaptive learning rate $\alpha_t$ and by using the standard GM-PHD-AA fusion approach, respectively}, have been given in Fig. \ref{fig:HomoPHDcomp}. It shows that the OSPA error has been significantly reduced by the proposed GM-PHD fit approach {especially when $\alpha_1=0.4, \beta_t=0.6$} although its reduction is not so significant as the standard GM-PHD-AA fusion does. {However, the GM-weight-fit approach using fixed learning rate $\alpha_t=0.2$ will diverge after $t \geq 4$ but not in the case using adaptive learning rate, although the latter slows down the convergence somehow and results in comparably less OSPA error reduction.}


{Similarly, we also consider the homogeneous MB filter fusion by using the proposed GM-weight-fit approach for PHD consensus based on either fixed or adaptive learning rate $\alpha_t$ and by using the standard GM-PHD-AA fusion approach \cite{Li17PC}, respectively, in comparison with the standard MB-AA filter fusion based on the B2B association and inter-sensor GM exchange that enable the parallel Bernoulli-AA fusion, where the B2B association is carried out via the optimal assignment based on the Hungarian algorithm or clustering as addressed in \cite{Li20AAmb}. The average OSPA errors are given in Fig. \ref{fig:HomoMBcomp}. The results have demonstrated the effectiveness of the proposed, approximate GM-PHD-AA fit approach for fusing these MB filters especially in case of suitable learning rate and fading rate, although the fusion gain is much lower than the standard MB-AA fusion. Furthermore, we consider the LMB filters for which the relevant average OSPA results are given in Fig. \ref{fig:HomoLMBcomp}. Compared with the PHD and even the MB filters, the fused LMB filters are improved by the GM-weight-fit approach for PHD consensus much less and is prone to over-fit. It will even deteriorate when $t>2$ in the case of using fixed $\alpha_t=0.2$ and even when $t>1$ in case of over-large adaptive learning rate such as $\alpha_t=0.4$.}

{We conjugate that better results can be expected in the case of heterogeneous MB and LMB fusion when {an appropriate B2B association procedure is employed so that the GM-weight-fit can be performed with regard to each matched/associated BC.} To this end, the labels can be removed temporarily so that the LMBs reduce to MBs, 
which enables BC association/clustering strategies as given for the MB-AA fusion \cite{Li20AAmb}. 
Overall, the OSPA reduction in the case of MB and LMB filters is lower than that for the PHD filters---as noted in remark \ref{remark:insufficientPHD1}--- and over-fit is easier to occur in the case of using fixed learning rate. Over-fit (non-convergence) has been significantly reduced or even averted by using proper, gradually-decreasing learning rate. These results confirm the analysis given in Section \ref{sec:insufficiencyofB2B}.}

\begin{figure}
  \centering
  \includegraphics[width=8cm]{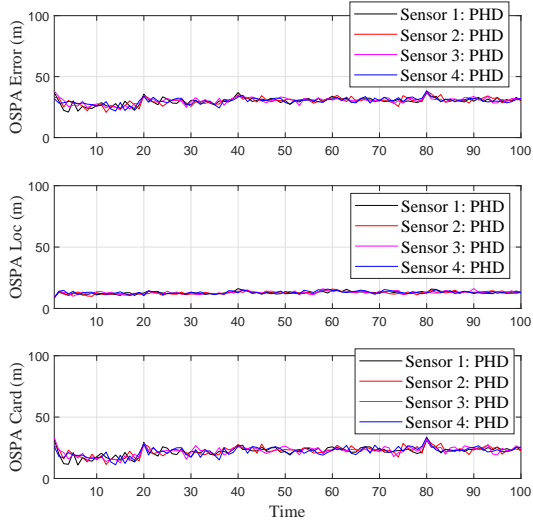}\\
  \caption{{The OSPA error, cardinality error and localization error of each local PHD filter over time based on only one GM-weight-fit iteration ($\alpha_1=0.2$) for PHD consensus}} \label{fig:PHD-opsa-I=1}
  \vspace{-2mm}
\end{figure}

\begin{figure}
  \centering
  \includegraphics[width=8cm]{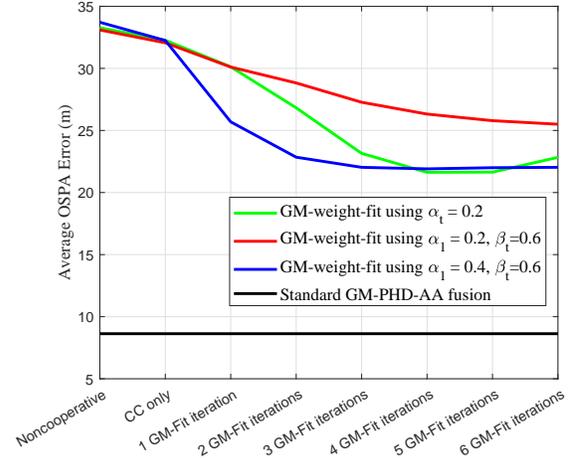}\\
  \caption{{Average OSPA errors of four PHD filters based on GM-weight-fit for PHD consensus using either fixed or adaptive learning rate in comparison with the standard GM-PHD-AA fusion approach given in \cite{Li17PC}}} \label{fig:HomoPHDcomp}
  \vspace{-2mm}
\end{figure}


\begin{figure}
  \centering
  \includegraphics[width=8cm]{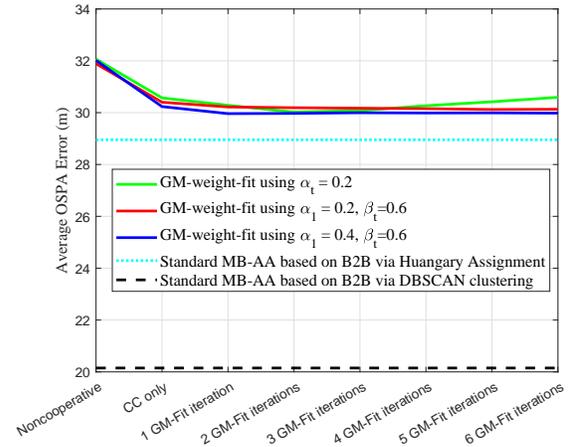}\\
  \caption{{Average OSPA errors of four MB filters based on GM-weight-fit for PHD consensus using either fixed or adaptive learning rate in comparison with the standard MB-AA filters \cite{Li20AAmb} based on two different B2B association strategies respectively}} \label{fig:HomoMBcomp}
  \vspace{-2mm}
\end{figure}

\begin{figure}
  \centering
  \includegraphics[width=8cm]{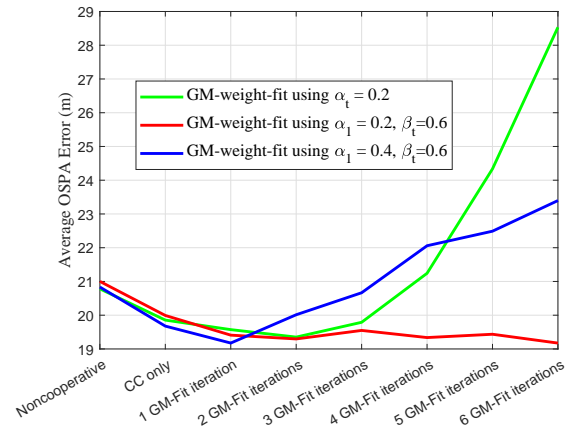}\\
  \caption{{Average OSPA errors of four LMB filter based on GM-weight-fit for PHD consensus using either fixed or adaptive learning rate}} \label{fig:HomoLMBcomp}
  \vspace{-2mm}
\end{figure}

\subsection{Heterogeneous PHD, MB and LMB fusion}
This section tests the performance of the proposed GM-AA fit approach in the heterogeneous case. The average results of the proposed GM-AA fit approach for two PHD and two MB filters cooperation, for two PHD and two LMB filters and for two-MB and two-LMB filters cooperation, {all filters using adaptive learning rate $\alpha_1=0.2, \beta_t=0.6$,} are given in Figs. \ref{fig:HeterPHDMB}, \ref{fig:HeterPHDLMB} and \ref{fig:HeterMBLMB}, respectively. {As shown, the heterogeneous fusion reduces the OSPA error of the PHD filters and even the MB filters for which more GM-weight-fit iterations lead to greater reduction of the OSPA error. This confirms the effectiveness of the proposed approach for heterogeneous RFS filter fusion. Similar to the case of homogeneous fusion in section \ref{sec:homogeneousSim}, the proposed GM-AA fit approach improves the PHD filter more than the MB/LMB filter {when they are fused with each other}. However, surprisingly, 
the heterogenous fusion with the MB filter does not benefit the LMB filter more than the CC-only strategy, which means the proposed sequential GM-weight-fit does not improve the LMB filter except the CC part. This is mainly due to the lack of B2B association as addressed in Section \ref{sec:insufficiencyofB2B}. To address it, the fusion needs to be carried out in a track-wise fashion for which appropriate B2B association across sensors is required. We leave this to the future work. } 

\begin{figure}
  \centering
  \includegraphics[width=8cm]{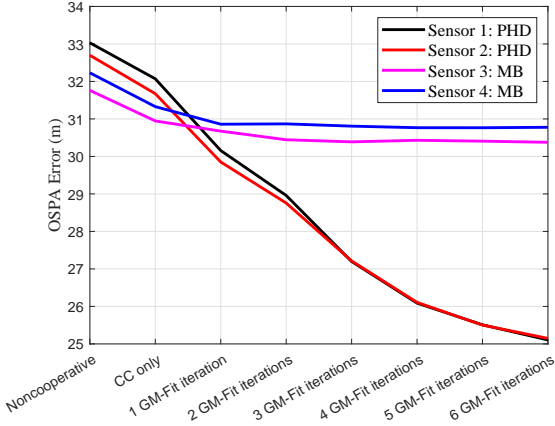}\\
  \caption{{The OSPA errors of two PHD and two MB filters against different levels of GM-weight-fit using adaptive learning rate $\alpha_1=0.2, \beta_t=0.6$}} \label{fig:HeterPHDMB}
  \vspace{-2mm}
\end{figure}

\begin{figure}
  \centering
  \includegraphics[width=8cm]{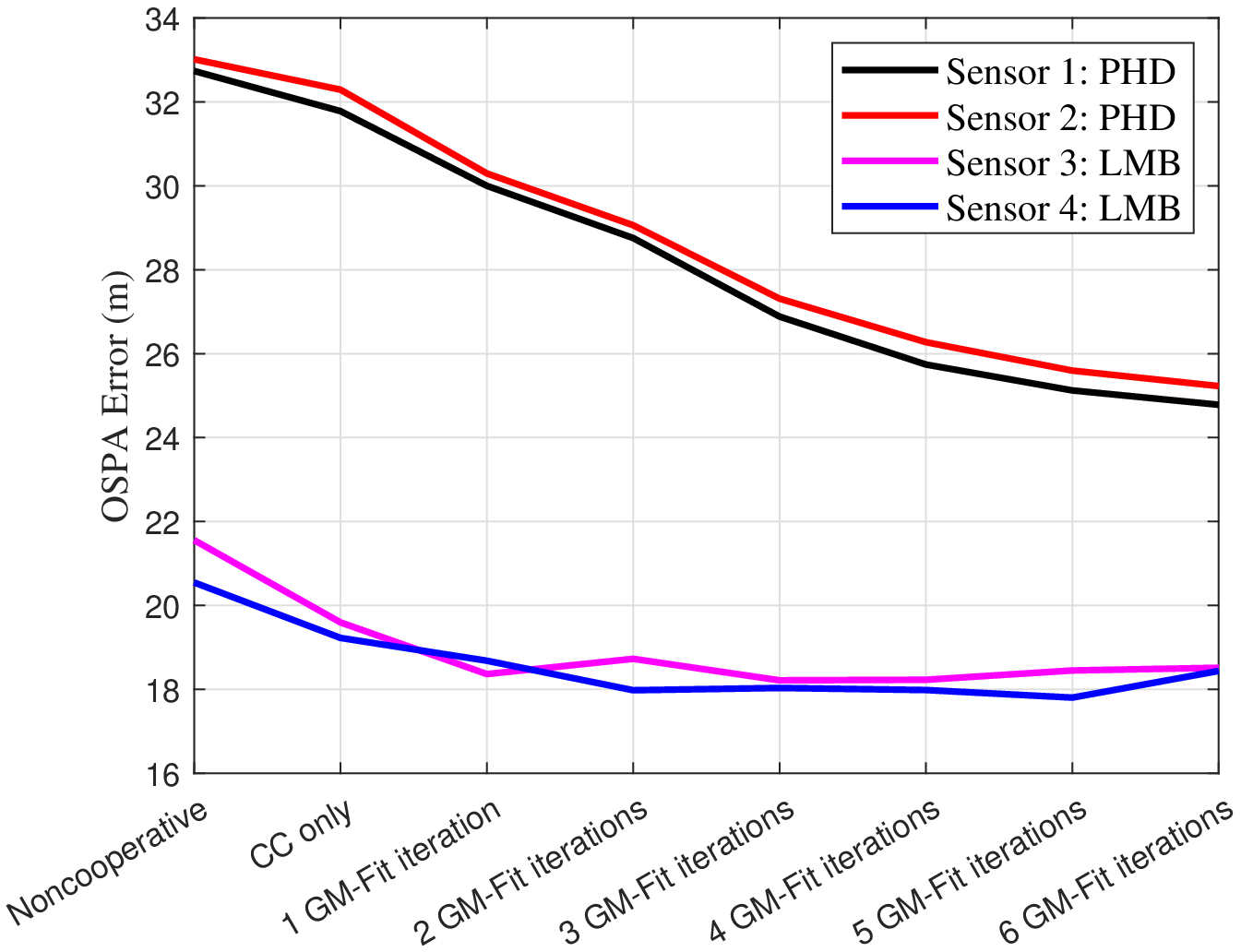}\\
  \caption{{The OSPA errors of two PHD and two LMB filters against different levels of GM-weight-fit using adaptive learning rate $\alpha_1=0.2, \beta_t=0.6$}} \label{fig:HeterPHDLMB}
  \vspace{-2mm}
\end{figure}

\begin{figure}
  \centering
  \includegraphics[width=8cm]{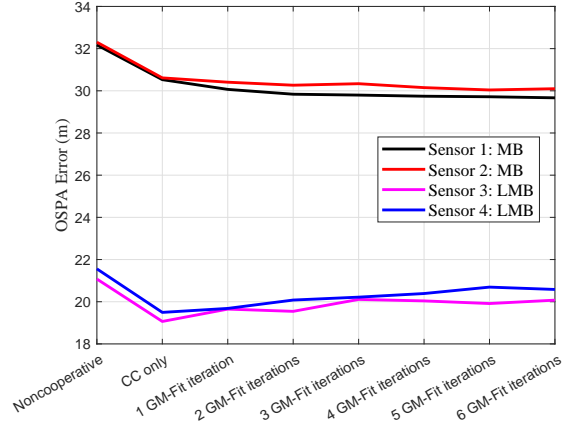}\\
  \caption{{The OSPA errors of two MB and two LMB filters against different levels of GM-weight-fit using adaptive learning rate $\alpha_1=0.2, \beta_t=0.6$}} \label{fig:HeterMBLMB}
  \vspace{-2mm}
\end{figure}

We further consider a more complicated configuration of two PHD filters, one MB filter and one LMB filter. 
The average OSPA error of each filter based on the proposed GM-weight-fit approach {for PHD consensus} is given in Fig. \ref{fig:linearPHDMBLMB}, which confirms that the heterogeneous unlabeled PHD-AA fusion based on GM-weight-fit improve all filters, especially the PHD filters. 
However, {it remains open how to best set the learning rate, the fading rate as well as the number of GM-weight-fit iterations for each filter in different cooperation configures. According to our experience, a rule of thumb choice is $\alpha_1=0.2, \beta_t=0.6$.} 

{Moreover, the average computing times costed by one iteration of the PHD, MB, LMB filtering operation separately and by their fusion operation with different numbers of GM-weight-fit iterations are also given in Fig.~\ref{fig:linearPHDMBLMB}. As shown, the proposed GM-weight-fit approach is computing more costly as compared with the filtering operations in general due to the intensive computation for calculating the correlation among all GCs as in \eqref{eq:minISD-wi}. However, the computation time required increase slowly with the increase of the number of GM-weight-fit iterations. This is because the correlation between GCs only needs to be calculated once (and be recorded for multiple uses in different iterations) because the means and variances of the GCs are unchanged. } 

{Finally, we test the heterogeneous fusion of two PHD filters, one MB filter and one LMB filter based on the nonlinear measurement model \eqref{eq:RangeBearing} for which the local filters employ the unscented approximation for dealing with the nonlinearity. The average OSPA errors and average computing times for one filtering iteration of each filter and for their heterogeneous fusion are given in Fig. \ref{fig:nonlinearPHDMBLMB}, respectively. The results comply with those in the linear measurement case. Differently, the two PHD filters that are localized at different positions which are related to their respective range-bearing measurements perform quite differently with each other, although both of them are improved by the heterogeneous fusion more significantly as compared with the MB and LMB filters.}

\begin{figure}
  \centering
  \includegraphics[width=9cm]{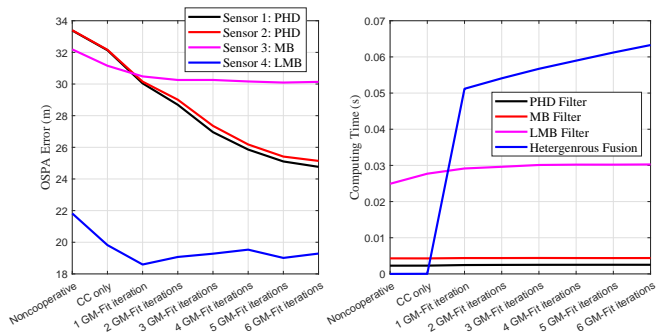}\\
  \caption{{Linear measurement case: the OSPA errors and average computing time for each filtering step of two PHD, one MB and one LMB filters, and their heterogeneous fusion time against different levels of GM-weight-fit using adaptive learning rate $\alpha_1=0.2, \beta_t=0.6$}} \label{fig:linearPHDMBLMB}
  \vspace{-2mm}
\end{figure}


\begin{figure}
  \centering
  \includegraphics[width=9cm]{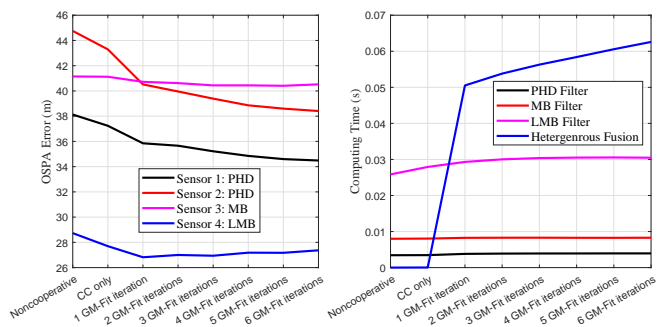}\\
  \caption{{Nonlinear measurement case: the OSPA errors and average computing time for each filtering step of two PHD, one MB and one LMB filters, and their heterogeneous fusion time against different levels of GM-weight-fit using adaptive learning rate $\alpha_1=0.2, \beta_t=0.6$}} \label{fig:nonlinearPHDMBLMB}
  \vspace{-2mm}
\end{figure}

\section{Conclusion and Future Work} \label{sec:conclusion}
We propose a heterogenous unlabeled and labeled RFS filter fusion approach which averages the unlabeled PHDs {of the diverse RFS filters} based on the GM implementation.
A computationally efficient, approximate approach to GM weight fit is proposed which sequentially revises only the weights of {the Gaussian components of the local GM in multiple iterations for PHD consensus. Both fixed and adaptive parameters have been suggested for controlling the convergence of the proposed fitting approach}. Simulations have demonstrated the effectiveness of the proposed approach {using adaptive fitting parameters} for both homogeneous and heterogeneous PHD-MB-LMB filter fusion. Our approach improves the PHD filter significantly by cooperating with the MB/LMB filters or the other PHD filters. In comparison, it improves the MB and LMB filters less. The reasons are twofold. 
First, the PHD is only the first order moment of the MB/LMB density for which the PHD-consensus is insufficient for MB/LMB consensus. Second, the proposed heterogenous fusion is carried out in the unlabeled domain which applies no track/label matching for track-wise fusion. This leaves great space for further improvement.

Improvement can be expected if more parameters of the L-GCs such as the means and covariances can be jointly optimized with the weights for better GM-fit and if the BCs can be properly associated (or the labels can be properly matched) in the case of MB/LMB filter for B2B-based PHD fusion. {To this end, advanced approximation/optimization methods such as variational Bayes and expectation maximization can be considered.} 
However, label/track matching/association remains an open-ended, challenging issue even for homogeneous MB/LMB filter fusion and needs further investigation. {Valuable future works also include addressing the partially overlapping field of view of the sensors, unknown target birth/dynamics/death models, unknown sensor profiles (such as the target detection probability and clutter rate) and so on based on the framework of the heterogeneous fusion.} 


\appendix
\subsection{Partial differential of the ISD of GMs} \label{appendix:iseMIN}
The ISD between two GMs $p(\mathbf{x}) = \sum_{n \in \mathbf{I}_1}\alpha_n \mathcal{N}(\mathbf{x}; {\bm{\mu}_n},{\mathbf{P}_n}),
  q(\mathbf{x}) = \sum_{m \in \mathbf{I}_2}\beta_m\mathcal{N}(\mathbf{x}; {\mathbf{m}_m},{\mathbf{S}_m})$ is given as follows {\cite{Williams06}}
\begin{align}
\text{ISD}\big(p\|q\big)
 \triangleq & \int{\big( p(\mathbf{x}) - q(\mathbf{x}) \big)^2 d \mathbf{x}} \\
& = \text{ISD}_{\alpha} + \text{ISD}_{\beta} -2\text{ISD}_{\alpha \beta}
\end{align}
where
\begin{align}
\text{ISD}_\alpha & = \sum\limits_{n,n'\in \mathbf{I}_1} {{\alpha_n}{\alpha _{n'}}\mathcal{N}({\bm{\mu}_n}; {\bm{\mu} _{n'}},{\mathbf{P}_n}+ {\mathbf{P} _{n'}})}
 \label{eq:alphaISD}\\
\text{ISD}_\beta & = \sum\limits_{m,m'\in \mathbf{I}_2} {{\beta_m}{\beta _{m'}}\mathcal{N}({\mathbf{m}_m}; {\mathbf{m}_{m'}},{\mathbf{S}_m} + {\mathbf{S}_{m'}})}  \label{eq:BetaISD}\\
\text{ISD}_{\alpha \beta} & =  \sum\limits_{n\in \mathbf{I}_1,m\in \mathbf{I}_2} {{\alpha_n}{\beta_m}\mathcal{N}({\bm{\mu}_n}; {\mathbf{m}_m},{\mathbf{P}_n}{\rm{ + }}{\mathbf{S}_m})} \label{eq:alphaBetaISD}
\end{align}

Here, it is straightforward to derive that
\begin{align}
  \frac{\partial  \text{ISD}\big(p\|q\big) }{\partial \alpha_n} = & 2\sum\limits_{n' \neq n, n' \in \mathbf{I}_1} {{\alpha _{n'}}\mathcal{N}({\bm{\mu}_n};{\bm{\mu} _{n'}},{\mathbf{P}_n}+ {\mathbf{P} _{n'}})}
\nonumber\\
 & + {2{\alpha _{n}}\mathcal{N}({\bm{\mu}_n}; {\bm{\mu} _{n}},2{\mathbf{P}_n} )} \nonumber\\
 & - 2\sum\limits_{m\in \mathbf{I}_2} {{\beta_m}\mathcal{N}({\bm{\mu}_n};{\mathbf{m}_m},{\mathbf{P}_n}{\rm{ + }}{\mathbf{S}_m})} \label{eq:derISD}
 \end{align}
 \begin{align}
 \frac{\partial^2 \text{ISD}\big(p\|q\big) }{\partial \alpha_n^2}
 & = {2 \mathcal{N}({\bm{\mu}_n};{\bm{\mu}_n}, 2{\mathbf{P}_n})} \nonumber \\
 & = \frac{2}{{{\left| {{2\bf{P}}_n } \right|}^{1/2}} { (2\pi )}^{d/2} } \label{eq:TwicederISD} \\
 & >0
\end{align}
Setting \eqref{eq:derISD} 
zero will lead to
\begin{align}
\alpha_n & =
  {{\left| {{2\bf{P}}_n } \right|}^{1/2}} { (2\pi )}^{d/2}\sum\limits_{m \in \mathbf{I}_2} {{\beta_m}\mathcal{N}({\bm{\mu}_n};{\mathbf{m}_m},{\mathbf{P}_n}{\rm{ + }}{\mathbf{S}_m})} \nonumber \\
&- {{\left| {{2\bf{P}}_n } \right|}^{1/2}} { (2\pi )}^{d/2}\sum\limits_{n' \neq n} {{\alpha _{n'}}\mathcal{N}({\bm{\mu}_n};{\bm{\mu} _{n'}},{\mathbf{P}_n} \rm{+} {\mathbf{P} _{n'}})}  \label{eq:minISD-a}
\end{align}
which yields the minimum ISD given all the other L-GC weights $\{\alpha_{n'}\}_{n' \neq n, n' \in \mathbf{I}_1}, \{\beta\}_{m \in \mathbf{I}_2}$.

\section*{Acknowledgement}
The authors would like to thank the anonymous reviewers for their constructive and detailed comments, which motivated the idea of designing the adaptive learning rate in the proposed GM-weight-fit approach for PHD consensus.

\bibliographystyle{IEEEtran}
\bibliography{HeterogenousFusion}


\end{document}